\documentclass[sigconf]{aamas} 

\usepackage{balance} 
\usepackage{amsmath}

\usepackage{amssymb}
\usepackage{algorithm}
\usepackage[noend]{algpseudocode}
\usepackage{varwidth}
\usepackage{ulem}
\usepackage{pifont}
\usepackage{amsthm}
\newtheorem{theorem}{Theorem}
\newtheorem{definition}{Definition}
\usepackage{tcolorbox}  
\definecolor{darkorange}{RGB}{255, 140, 0}
\definecolor{darkblue}{RGB}{84, 112, 198}
\definecolor{lightgreen}{RGB}{145, 204, 117}
\definecolor{lightyellow}{RGB}{250, 200, 88}
\definecolor{lightred}{RGB}{238, 102, 102}
\definecolor{lightblue}{RGB}{115, 192, 222}

\newtcolorbox{promptbox}[2][Prompt]{
colback=black!5!white,
arc=3pt, 
boxrule=0.5pt,
fonttitle=\bfseries,
title=#1, 
before upper={\small}, fontupper=\fontfamily{ptm}\selectfont,
colframe=#2, 
}

\copyrightyear{2026}
\acmYear{2026}
\acmDOI{}
\acmPrice{}
\acmISBN{}

\acmSubmissionID{56}

\title[AAMAS-2026 Formatting Instructions]{Simulation-Free PSRO: Removing Game Simulation from Policy Space Response Oracles}

\author{Yingzhuo Liu$^{*}$}
\affiliation{
  \institution{BUPT}
  \city{Beijing}
  \country{China}
}
\email{liuyingzhuo86@bupt.edu.cn}

\author{Shuodi	Liu$^{*}$}
\affiliation{
	\institution{BUPT}
	\city{Beijing}
	\country{China}
}
\email{liushuodi@bupt.edu.cn}

\author{Weijun Luo}
\affiliation{
	\institution{BUPT}
	\city{Beijing}
	\country{China}
}
\email{wjl@bupt.edu.cn}

\author{Liuyu Xiang}
\affiliation{
	\institution{BUPT}
	\city{Beijing}
	\country{China}
}
\email{xiangly@bupt.edu.cn}

\author{Zhaofeng He$^{\dagger}$}
\affiliation{
	\institution{BUPT}
	\city{Beijing}
	\country{China}
}
\email{zhaofenghe@bupt.edu.cn}

\begin{abstract}
Policy Space Response Oracles (PSRO) combines game-theoretic equilibrium computation with learning and is effective in approximating Nash Equilibrium in zero-sum games. However, the computational cost of PSRO has become a significant limitation to its practical application. Our analysis shows that game simulation is the primary bottleneck in PSRO's runtime. To address this issue, we conclude the concept of Simulation-Free PSRO and summarize existing methods that instantiate this concept. Additionally, we propose a novel Dynamic Window-based Simulation-Free PSRO, which introduces the concept of a strategy window to replace the original strategy set maintained in PSRO. The number of strategies in the strategy window is limited, thereby simplifying opponent strategy selection and improving the robustness of the best response. Moreover, we use Nash Clustering to select the strategy to be eliminated, ensuring that the number of strategies within the strategy window is effectively limited. Our experiments across various environments demonstrate that the Dynamic Window mechanism significantly reduces exploitability compared to existing methods, while also exhibiting excellent compatibility. Our code is available at https://github.com/enochliu98/SF-PSRO.
\end{abstract}


\newcommand{\BibTeX}{\rm B\kern-.05em{\sc i\kern-.025em b}\kern-.08em\TeX}

\begin{document}


\pagestyle{fancy}
\fancyhead{}


\maketitle 



\section{Introduction}

\begin{figure}[htbp]  
    \centering
    \includegraphics[width=\columnwidth]{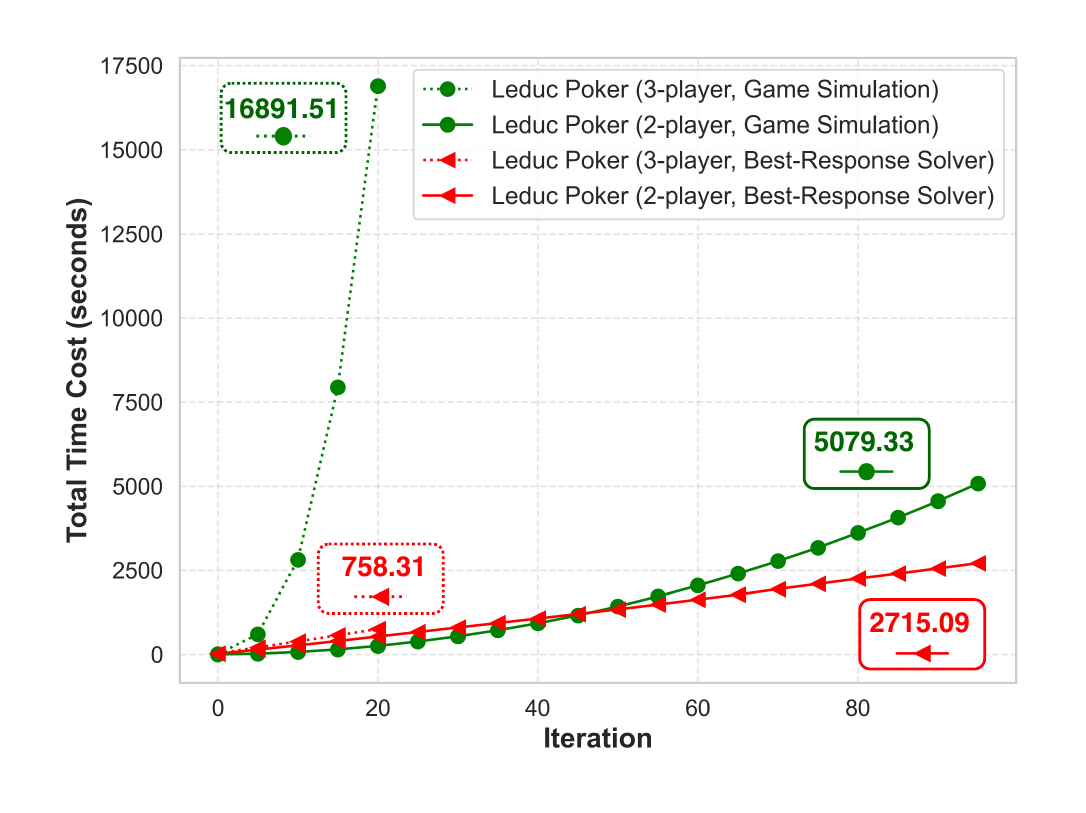}
    \caption{Game Simulation is the primary bottleneck of PSRO in terms of computational cost}
    \label{f1}
\end{figure}


In recent years, Multi-Agent Systems (MAS) \cite{van2008multi} have become an emerging research focus. A MAS can typically be viewed as a game, comprising multiple decision-making agents that interact within a shared environment. As a result, the optimal behavior of an agent in such systems often depends on the behavior of other agents. To understand the strategic behavior among these agents, game theory \cite{simon1945theory} provides solution concepts like Nash Equilibrium (NE) \cite{fudenberg1991game}. However, even in relatively simple zero-sum settings, when the number of players and strategies grows significantly, even polynomial-time methods like Linear Programming become computationally intractable. As a result, learning methods, such as Multi-Agent Reinforcement Learning (MARL) \cite{zhang2021multi, rashid2020monotonic,gu2021multi}, have been proposed as alternatives to traditional equilibrium computation methods. However, these methods face two major challenges: non-stationary \cite{zhang2021multi} (i.e., each agent faces a potentially moving target) and non-transitivity \cite{tang2025distributed} (i.e., there is no clear "better" strategy for an agent). The Policy Space Response Oracles (PSRO) \cite{lanctot2017unified, mcaleer2020pipeline, smith2023strategic} framework emerged as a natural combination, integrating traditional game-theoretic equilibrium computation with learning \cite{bighashdel2024policy}.

PSRO maintains a strategy set for each player and iteratively expands it by learning new strategies. The learning process involves three components \cite{bighashdel2024policy}: Game Simulation (GS), Meta-Strategy Solver (MSS), and Best-Response Solver (BRS). Specifically, utilities for all strategy profiles in the joint strategy space $X$ (i.e., each player maintains a local strategy set $X_i$, and $X =  {\textstyle \prod_{i}X_i} $) are evaluated via GS to construct the meta-payoff matrix $P$. Based on this matrix, the MSS computes a meta-strategy $\sigma \in \bigtriangleup X$, which serves as the best-response target for the next iteration. Then, each player independently computes a best-response (BR) $\beta_i$ to its corresponding best-response target $\sigma_{-i}$ (i.e., the meta strategy of all players except $i$) based on response objective. In standard PSRO, the response objective can be written as $RO_i(\sigma) = u_i(\beta_i, \sigma_{-i})$ and maximizing it over $\beta_i$ gives player $i$ a best-response against other players' strategies $\sigma_{-i}$. The resulting best-response $\beta_i$ is subsequently added to the player’s local strategy set $X_i$.

As described above, learning each new strategy requires executing all three steps. Moreover, in many games, PSRO typically needs to learn a large number of new strategies to sufficiently expand the strategy sets (e.g., in Leduc Poker, existing methods often run over 100 iterations). Consequently, the PSRO procedure is generally time-consuming. The above analysis reveals that the runtime of PSRO is primarily influenced by two factors: \ding{182} the number of strategies to be learned, and \ding{183} the time required to learn each strategy. In this work, we focus on the latter. We further analyze the runtime of each of the three steps involved in learning a new strategy (It is worth noting that we consider GS and MSS together in our analysis. This is because the purpose of GS is to construct the meta-payoff matrix, which is then used for computing meta-strategy). 

\begin{promptbox}[Conclusion from Fig.1]{black}

\textbf{Game Simulation is the primary bottleneck of PSRO in terms of computational cost, significantly outweighing the cost of the Best-Response Solver. This issue becomes even more severe as the number of players or the number of iterations increases.}

\end{promptbox}

This trend is primarily due to the fact that, in each iteration, the number of games simulated by GS increases with both the number of players and iterations (Further details are provided in Section 2). In contrast, the BRS typically performs a fixed number of training steps per iteration, making its runtime relatively insensitive to changes in the number of players or iterations.

However, in many cases, we do not have sufficient time and computational resources to train PSRO. Therefore, finding a time-efficient method that still delivers acceptable results is an urgent problem to address. Considering the time consumption of GS, we conclude the concept of Simulation-Free PSRO (SF-PSRO), which refers to a PSRO method that does not require GS \footnote{It should be clarified that by “Simulation-Free” we specifically mean the removal of the Game Simulation (GS) component from PSRO, not the elimination of all forms of simulation. For instance, simulations remain indispensable during the BR computation process.}. Currently, existing PSRO variants that support simulation-free remain limited, and they can be summarized from two perspectives: MSS and BRS. The former mainly focuses on how to find an appropriate meta-strategy without relying on the meta-payoff matrix. Existing methods \cite{mcaleer2022anytime, mcaleer2024toward, zhou2022efficient} typically use minimum-regret constrained profiles (MRCP)\cite{jordan2010strategy} to construct the MSS. The latter focuses on how to construct diversity objectives for the BRS without relying on the meta-payoff matrix. Current approaches mainly explore this from the perspective of behavioral diversity (BD) \cite{liu2021towards, yao2023policy}.

Considering that BRS often uses reinforcement learning (RL) \cite{mnih2015human, schulman2017proximal, lillicrap2015continuous} for training, which requires interaction with the environment, this interaction process is quite similar to the GS process. We wonder whether the information obtained during the BRS process could, to some extent, replace the GS process. Therefore, we propose a new PSRO based on a dynamic window, which replaces the strategy set maintained in previous PSRO. Specifically, we maintain a fixed-size window, and when the number of strategies reaches the upper limit of the window, a new strategy is added by eliminating an old one. A fixed-size window can effectively limit the number of strategies in the strategy set, thereby simplifying opponent strategy selection (optimizing MSS) and enhancing the performance of the best response against the strategies within the strategy set (optimizing BRS).

Furthermore, to identify the strategy to be eliminated, we record the outcomes of interactions between the best-response and best-response target during the BRS process to construct a sketchy meta-payoff matrix. Since the BR is continuously optimized throughout training, the recorded outcomes are inherently imprecise. As a result, the corresponding meta-payoff matrix is considered "sketchy". Nevertheless, this matrix is sufficiently informative to assist in selecting an underperforming strategy for removal. Specifically, one straightforward approach is to eliminate the strategy with the lowest average payoff in the sketchy meta-payoff matrix. However, this may be misleading due to countering relationships among strategies. For example, a strategy with a low average payoff might strongly counter a dominant strategy while being exploited by several weaker ones, thus still holding strategic value. To address this, we construct Nash Clusterings \cite{czarnecki2020real} based on the sketchy meta-payoff matrix and select the weakest strategy from the lowest-performing Nash cluster for removal.

To summarize, our contributions are as follows:

\begin{itemize}
    \item We analyze the time consumption of each component in PSRO and point out that the cost of game simulation is increasingly severe as the number of players and iterations increases.
    \item We conclude the concept of Simulation-Free PSRO and summarize methods supporting simulation-free from the perspectives of meta-strategy solver and best-response solver.
    \item We propose a Dynamic Window mechanism, which can be effectively integrated with existing SF-PSRO methods.
\end{itemize}


\section{Simulation-Free PSRO}

\label{s2}

As previously concluded, "\textbf{\textit{Game simulation is the primary bottleneck of PSRO in terms of computational cost, significantly outweighing the cost of the Best-Response. This issue becomes even more severe as the number of players or the number of iterations increases.}}" A concrete example can further illustrate this point.

Consider an $N$-player game in which each player maintains $M$ strategies. The total number of strategy profiles is then $M^N$. To construct the corresponding meta-payoff matrix, simulations must be conducted for each strategy profile. Assuming that K simulations are required per profile to obtain reliable estimates, the total number of simulations becomes $(M^N - (M-1)^N) \times K$ (In PSRO, the strategy set is iteratively expanded, and only the missing entries in the meta-game payoff matrix need to be filled). Since $N$ corresponds to the number of players and $M$ is determined by the number of PSRO iterations, both large M and N can lead to an exponential explosion in simulation cost, thereby validating our earlier conclusion. Moreover, to ensure accurate evaluation, the value of $K$ is typically set to a large number. For instance, in a 4-player game where each player maintains 10 strategies and $K$ is set to 1000, the total number of simulations required to build the meta-game is $(10^4 - 9^4) \times 10^3 \approx 3.439 \times 10^6$. In contrast, generating a single best-response usually requires only around $10^4$ simulations. This significant imbalance results in considerable inefficiency, which severely impacts the overall training time.

As illustrated in Fig. \hyperref[f2]{2}, the Vanilla PSRO \cite{lanctot2017unified} consists of three key components: game simulation, meta-strategy solver, and best-response solver. To accelerate the execution of PSRO, we eliminate the GS process, leading to the development of Simulation-Free PSRO. We observe that several existing PSRO variants (or some components in these variants) already support simulation-free. In the following, we summarize these methods from two perspectives: MSS and BRS.

\begin{figure*}[htbp]  
    \centering
    \includegraphics[width=0.95\textwidth]{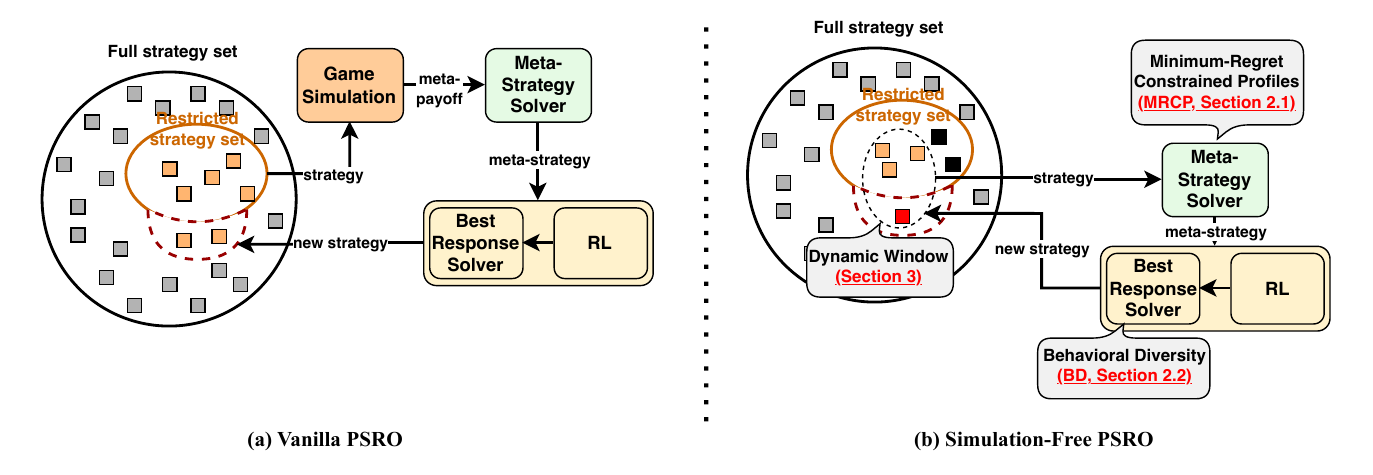}
    \caption{Comparison between Vanilla PSRO and Simulation-Free PSRO}
    \label{f2}
\end{figure*}

\noindent
\noindent
\begin{minipage}{\columnwidth}
  \begin{algorithm}[H]
    \caption{Vanilla PSRO}
    \begin{algorithmic}[1]
        \State \textbf{Input:} initial strategy sets $X=(X_1, X_2)$
        \State \textbf{Input:} initial meta-payoff matrix $P^X$
        \State \textbf{Input:} initial meta-strategies $ \sigma_i$
        \While{not terminated}
            \For{player $i\in\{1,2\}$}
                \For{many episodes}
                \State Train best response $\beta_i$ against \\  \hspace{\algorithmicindent}\hspace{\algorithmicindent}\hspace{\algorithmicindent} $\beta_{-i} \sim \sigma_{-i}$
                \EndFor
            \State $X_i = X_i \cup \{\beta_i\}$
            \EndFor
        \State \textcolor{blue}{\underline{Run game simulation to compute $P^X$}}
        \State Compute a meta-strategy $\sigma$ from $P^X$
        \EndWhile
    \State \textbf{Return:} $\sigma$
    \end{algorithmic}
  \end{algorithm}
\end{minipage}
\hfill
\begin{minipage}{\columnwidth}
  \begin{algorithm}[H]
    \caption{Simulation-Free PSRO}
    \begin{algorithmic}[1]
        \State \textbf{Input:} initial strategy windows $X^w=(X^w_1, X^w_2)$
        \State \textbf{Input:} initial sketchy meta-payoff matrix $P_{s}^{X}$
        \While{not terminated}
            \State Initialize the meta-strategy $\sigma_i$ to uniform over \\
            \hspace{\algorithmicindent}$X^w_i$ for $i \in \{1,2\}$
            \For{$i \in \{1, 2\}$}
                \For{$n$ iterations}
                    \For{$m$ iterations}
                        \State \textcolor{red}{\# BD\footnote{BD: computing the best-response with the behavioral diversity regularization term}: via Equation (2) in Section 2.2}
                        \State \textcolor{red}{\underline{Update $\beta_{i}$ against $\sigma_{-i}$}}
                    \EndFor
                    \State \textcolor{red}{\# MRCP\footnote{MRCP: computing the meta-strategy with the MRCP through regret minimization}: via Equation (1) in Section 2.1}
                    \State \textcolor{red}{\underline{Update $\sigma_{-i}$ against $\beta_{i}$}}
                \EndFor
                \State $X^w_i \gets X^w_i \cup \{\beta_i\}$ for $i \in \{1,2\}$
                \State \textcolor{orange}{\underline{Update $X^w$ and $P^X_s$ via dynamic window in section 3}}
            \EndFor
        \EndWhile
    \State \textbf{Return:} $\sigma$
    \end{algorithmic}
  \end{algorithm}
\end{minipage}

\subsection{Meta-Strategy Solver in Simulation-Free PSRO}

MSS computes the meta-strategy based on the meta-payoff matrix, which serves as the target for the next BR. In this setting, GS for the meta-payoff matrix are only required to compute BR target (i.e., meta-strategy). Therefore, if the meta-strategy can be obtained through alternative approaches, there is no need to maintain the complete meta-payoff matrix, effectively avoiding redundant simulations. 

To achieve this goal, existing methods often adopt minimum-regret constrained profiles (MRCP) \cite{jordan2010strategy} as the MSS. An MRCP is defined as the profile with the minimum regret with respect to the full game.

Given that computing an MRCP is highly challenging, Anytime PSRO \cite{mcaleer2022anytime} primarily focuses on two-player zero-sum games, and computes MRCPs through regret minimization against a best-response. Within one iteration, (1) two restricted games are constructed, where one player is unrestricted; (2) for both players, a best-response is trained against the restricted distribution (i.e., meta-strategy), while the restricted distribution is updated via a no-regret algorithm against this BR; (3) the resulting BR is then added to the population. In other words, the computation of the meta-strategy and the best response is integrated into a unified process, where the meta-strategy is directly updated based on the outcomes of interactions between each strategy in the strategy set corresponding to the meta-strategy and the best response. The update rule of the regret minimization \cite{cesa2006prediction} for the meta-strategy is formulated as follows, 

\begin{equation}
    \sigma_{i} = \frac{exp(\eta \hat{S}_{i})}{ {\textstyle \sum_{j=1}^{k} exp(\eta \hat{S}_{j}}) }  \quad \mathrm{for \ each} \ \ i \in [1,...,k]  
\end{equation}

where $\eta$ denotes the learning rate, and $\hat{S}_i$ represents the average outcomes over the last 1000 episodes in which the given strategy was played.

Furthermore, the outcomes of interactions are obtained during the data collection process of BRS. Specifically, in each episode, an opponent strategy is sampled from the strategy set corresponding to the meta-strategy according to the meta-strategy, and the best response is trained against this sampled strategy. At the end of the episode, the outcome of the match between the best-response and the sampled strategy is recorded. The regret minimization algorithm then updates the meta-strategy based on these recorded results, without requiring any additional GS. Therefore, this approach serves as a representative example of SF-PSRO.


Self-Play PSRO \cite{mcaleer2024toward} extends Anytime PSRO by including not one but two strategies in the empirical game at each iteration: one is the best response to MRCP, and the other is the best response to the opponent's most recently added strategy (i.e., the strategy introduced in the previous PSRO iteration). Efficient PSRO \cite{zhou2022efficient} also leverages MRCP to avoid game simulation. In addition, it introduces parallelized best-response training and proposes a warm-start technique to address the re-training issue in MRCP. The detailed pseudocode for Anytime PSRO, Self-Play PSRO and Efficient PSRO can be found in the Appendix A.

\subsection{Best-Response Solver in Simulation-Free PSRO}

Each player independently computes a best-response to its corresponding best-response target based on response objective. In existing methods for optimizing response objective, most methods introduce a diversity regularization term into the response objective via Equation (2), which can be broadly categorized into Behavioral Diversity (BD) \cite{yao2023policy,liu2021towards} and Response Diversity (RD) \cite{balduzzi2019open, liu2022unified, perez2021modelling, liu2021towards}.

\begin{equation}
    \beta_i = \mathop{\arg\max}_{\beta_i} \{u_i(\beta_i, \sigma_{-i}) + \lambda * \mathrm{diversity}(\beta_i) \}
\end{equation}

where player $i$ computes a best-response $\beta_i$ to its corresponding best-response target $\sigma_{-i}$. The $\text{diversity}(\beta_i)$ denotes a diversity regularization term computed based on $\beta_i$, and $\lambda$ is its associated weight.

A key limitation of RD lies in their objective: enhancing diversity in RD aims to enlarge the gamescape, which does not necessarily correspond to closer to a full game NE. Moreover, RD relies on computations over the meta-payoff matrix. However, in SF-PSRO, no such matrix is maintained due to the absence of GS. In contrast, BD effectively address the limitations of RD. Specifically, BD ensures a bijective and linear mapping between representations and policies, thereby guaranteeing that the strategies generated via BD expansion enlarge the policy hull (PH) \cite{yao2023policy}, thus reducing the population exploitability (PE) \cite{yao2023policy} and corresponding to closer to a full game NE. Furthermore, the computation of BD does not rely on the meta-payoff matrix, making it well-suited for scenarios where SF-PSRO are used and game simulations are avoided. It is worth noting that existing BD-based PSRO variants still rely on game simulation and thus, strictly speaking, do not fall under SF-PSRO. However, the computation of the BD regularization term itself is simulation-free, making it applicable to other SF-PSRO variants.

Several BD-based methods have been proposed, existing diversity metrics explicitly or implicitly define a policy representation. In PSRO w. BD\&RD \cite{liu2021towards}, the joint occupancy measure is used to encode a policy, and BD is defined in the state-action space as the discrepancies of different strategies. However, f-divergence is typically used to measure the distance between two policies, but computing it based on occupancy measures is often intractable and usually approximated using neural network predictions. Instead, PSD-PSRO \cite{yao2023policy} uses the sequence-form representation and defines the policy distance using Bregman divergence, which can be simplified to a tractable form and optimized with state-action samples in practice.


\section{Methods}
\label{s31}
\subsection{Challenges with Existing Methods}

Most existing methods require maintaining a strategy set, which will be expanded iteratively. This leads to the following challenges:

\paragraph{Challenge 1: Optimizing MSS}
When the number of strategies in the strategy set becomes larger and larger, selecting an appropriate opponent (i.e., computing a meta-strategy via MSS) becomes increasingly difficult. However, the choice of opponent has a significant impact on the final performance of self-play training. In particular, in scenarios where GS is not conducted, the meta-payoff matrix is unavailable, and thus common heuristics must be used for opponent selection, such as Vanilla Self-Play \cite{heinrich2016deep} (selecting the most recently added strategy), Fictitious Self-Play \cite{heinrich2015fictitious} (assigning equal weights to all strategies in strategy set), or MRCP (setting weights based on regret minimization).

Specifically, in MRCP-based PSRO such as Anytime PSRO, the meta-strategy is updated based on the regret minimization, which relies on the best response’s performance against each strategy in the strategy set. Therefore, as the strategy set expands, a fixed number of training episodes must be distributed among more strategies, reducing the number of episodes allocated per strategy. This results in less accurate estimations of outcomes, which in turn degrades the effectiveness of the MRCP.

\paragraph{Challenge 2: Optimizing BRS}
When computing the best-response, a fixed number of training episodes is typically used. In each episode, a opponent strategy is selected from the strategy set according to the meta-strategy, and the best response is trained against this strategy. In other words, the fixed training episodes must be allocated across all opponent strategies based on the meta-strategy. As the opponent set grows, fewer episodes are allocated to each opponent strategy, which may cause the learned best response to underperform against certain opponent strategies.

\subsection{Dynamic Window}

\begin{figure*}[htbp]  
    \centering
    \includegraphics[width=\textwidth]{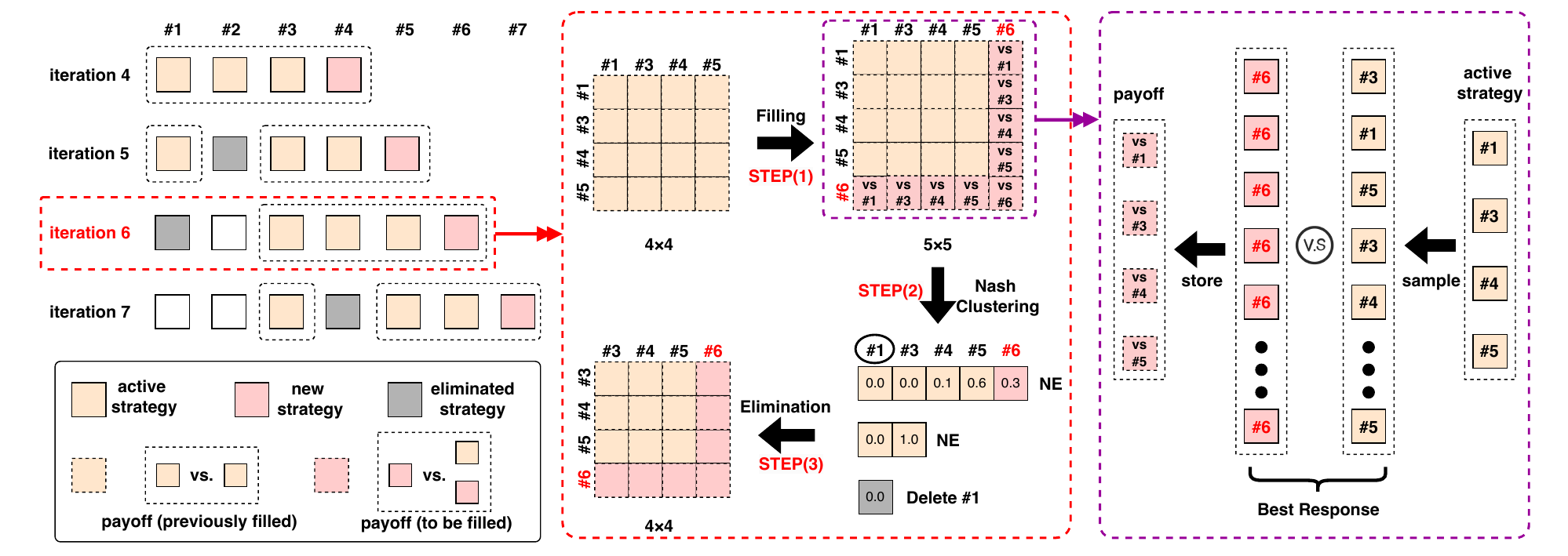}
    \caption{An illustration of the dynamic window mechanism. The left side shows each iteration, where a best-response (\textbf{new strategy}) is added to the strategy window, and one outdated strategy (\textbf{eliminated strategy}) is eliminated. The remaining ones are \textbf{active strategies}. The middle side depicts the three steps (\ding{182} Filling, \ding{183} Nash Clustering, and \ding{184} Elimination) used to identify the strategy to be eliminated. The right side details the implementation of the key step, Filling.}
    \label{f3}
\end{figure*}

Limiting the size of the strategy set offers a unified solution to both challenges. For \textit{Challenge 1}, we mitigate the difficulty in opponent selection by limiting the number of strategies in the strategy set (i.e., a smaller strategy set is filtered from the original strategy set). In this way, opponent strategies after filtering are generally more suitable as opponent strategies than those from the original strategy set. Moreover, limiting the size of the strategy set can also improve the accuracy of the estimation of outcomes in regret minimization. For \textit{Challenge 2}, a smaller strategy set ensures that each opponent strategy receives sufficient training during the BRS.

Furthermore, the key problem becomes how to effectively limit the size of the strategy set. Considering that the BRS is typically implemented via RL, which relies on interactions with the environment, a natural opportunity for strategy set filtering arises. Specifically, during training, an opponent strategy from the strategy set is selected at the beginning of each episode to compete against the current best response, and the outcome of the episode is recorded. After a sufficient number of episodes, statistically meaningful outcome information regarding the best response against various opponent strategies in the strategy set can be collected. However, existing methods fail to make full use of this information. In our view, one critical challenge lies in its staleness: as the best response is continuously optimized throughout training, the outcomes collected during earlier stages may no longer reflect the current relative performance. Although this issue prevents accurate ranking of all strategies in the strategy set, it is still possible to obtain a coarse-grained ordering, making it feasible to identify a relatively underperforming strategy.

Based on this insight, we propose a new SF-PSRO based on a dynamic window. A simple illustrative example is given in Fig. \hyperref[f3]{3}, we maintain a fixed-size window of size $N$ (for simplicity, we set $N=4$), and the strategies within this window are referred to as "active strategies". In each iteration, we train to obtain a new best-response and add it to the window. When the number of existing strategies in the window reaches $N$, for each new strategy added, we need to remove an eliminated strategy from the original window. 

Specifically, we maintain a sketchy meta-payoff matrix, which is constructed based on outcome information collected during the BRS process. The overall procedure for identifying an eliminated strategy is illustrated on the right side of Fig. \hyperref[f3]{3} and consists of three steps:

\textbf{(1) Filling}: The original sketchy meta-payoff matrix already contains the competitive outcomes among strategies \#1, \#3, \#4, and \#5. In the latest iteration, a new strategy (\#6) is trained via best-response solver, and its competitive outcomes (i.e., average return) against each of the active strategies in the current window (i.e., \#1, \#3, \#4, and \#5) are stored. By incorporating this information and exploiting the anti-symmetry of the payoff matrix, an updated matrix is formed.

\textbf{(2) Nash clustering}: Based on the updated matrix, to identify the single worst strategy, we apply Nash clustering \cite{czarnecki2020real} to construct multiple layers of Nash clusters, where each cluster contains a subset of strategies. These Nash clusters form a monotonic ordering with respect to Relative Population Performance (RPP). RPP is defined for two sets of agents $X_A$ and $X_B$, with a corresponding Nash equilibrium of the asymmetric game $(\sigma_A, \sigma_B) := \text{Nash}(P_{AB} \mid (A, B))$, as $\text{RPP}(X_A, X_B) = \sigma_A^\top P_{AB} \sigma_B$. Specifically, Nash clustering first computes the Nash equilibrium of the updated payoff matrix $P$ over the current set of strategies within the window (denoted as $\text{Nash}(P|X)$ when restricted to a strategy set $X$). The first cluster is then formed by collecting all strategies in the support of the equilibrium. This process is repeated on the remaining strategies until all strategies in the window have been assigned to a cluster.

    \begin{definition}
    Nash clustering $\mathcal{C}$ of the finite zero-sum symmetric game strategy set $X$ is defined by setting, for each $i \geq 1$: $N_{i+1} = \operatorname{supp}\left(\operatorname{Nash}\left(\mathbf{P} \big| X \setminus \bigcup_{j \leq i} N_j \right)\right)$ for $N_0 = \emptyset$ and $\mathcal{C} = (N_j : j \in \mathbb{N} \land N_j \neq \emptyset)$.
    \end{definition}
    Subsequently, we select the strategy with the smallest weight in the equilibrium corresponding to the last Nash cluster as the eliminated strategy.
    
\textbf{(3) Elimination}: The strategy identified for elimination is removed by deleting its corresponding row and column from the matrix.

The pseudocode for Vanilla PSRO and SF-PSRO is presented in Algorithms 1 and 2, respectively. Compared to SF-PSRO, Vanilla PSRO requires game simulation (line 10 in Algorithm 1). In SF-PSRO, we integrate the two major categories of methods compatible with SF-PSRO summarized in Section 2—specifically, BD (line 10 in Algorithm 2) and MRCP (line 12 in Algorithm 2). Our newly proposed dynamic window mechanism, as shown in line 14 of Algorithm 2, identifies strategies to be eliminated from the strategy window based on three steps: Filling, Nash clustering, and Elimination, and updates the sketchy meta-payoff matrix accordingly. This highlights the plug-and-play nature of our method, which allows for seamless integration with existing methods.

\section{Experiments}

In this section, we aim to experimentally investigate the following problems: 

\begin{itemize}
    \item \textbf{\textit{Is the Dynamic Window-based SF-PSRO effective? }}Specifically, can it achieve competitive performance while consuming significantly less time? We compare its performance against existing self-play methods (Vanilla Self-Play \cite{heinrich2016deep}, Fictitious Self-Play \cite{heinrich2015fictitious}) as well as state-of-the-art PSRO variants (Vanilla PSRO \cite{lanctot2017unified}, PSD-PSRO \cite{yao2023policy}, Anytime PSRO \cite{mcaleer2022anytime}). Among them, PSD-PSRO and Anytime PSRO are representative methods for BD and MRCP, respectively. Evaluations are conducted across a range of extensive games, including the relatively simple Leduc Poker \cite{LanctotEtAl2019OpenSpiel} and more complex games such as Goofspiel \cite{LanctotEtAl2019OpenSpiel}. To further validate the effectiveness of our proposed method in multiplayer settings, experiments are carried out in both Goofspiel (2-player) and Goofspiel (3-player). For Leduc Poker and Goofspiel (2-player), we measure and report the exploitability of the meta-NE throughout the training process. In GoofSpiel (3-player), computing exploitability is prohibitively expensive; instead, after all methods have completed, we evaluate their performance based on TrueSkill\cite{herbrich2006trueskill} and record the running time of each method. Additionally, we compare the total running time and the performance (optimal exploitability or TrueSkill) of each method. 
    \item \textbf{\textit{Are all components in Dynamic Window-based SF-PSRO effective?}} We first evaluate the impact of the two key components in the Dynamic Window mechanism: eliminated strategy selection (Nash clustering) and window size restrictions. Specifically, we compare the full method (Ours) with two ablated variants: one without the eliminated strategy selection (Ours\_w/o\_select, the eliminated strategy is randomly selected), and another without both components (Ours\_w/o\_select\_window, the window size is not restricted). In addition, we investigate the influence of window size on the final performance.
    \item \textit{\textbf{Can Dynamic Window mechanism be effectively integrated with existing SF-PSRO variants?}} To evaluate its compatibility, we incorporate Dynamic Window into BD and MRCP, and compare their performance with and without our Dynamic Window enhancement. 
\end{itemize}

Standard deviations in the results are computed over 5 independent runs. The implementation details of each game and method are provided in Appendix B.

\subsection{Main Results}

\begin{figure*}[htbp]  
    \centering
    \includegraphics[width=0.75\textwidth]{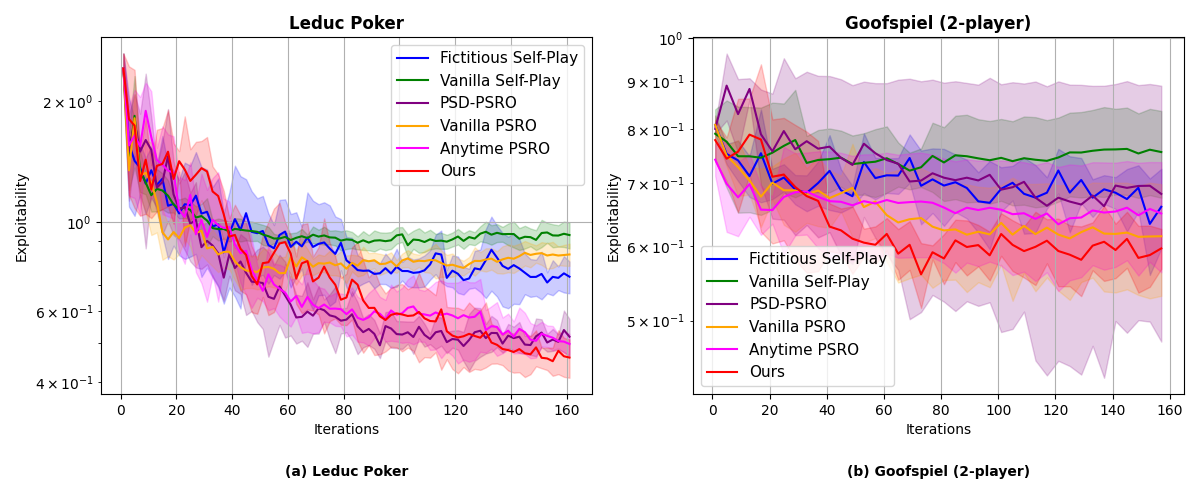}
    \caption{Exploitability of Leduc Poker and Goofspiel with 2e5 and 3e5 episodes for training BR}
    \label{f4}
\end{figure*}

\begin{figure*}[htbp]  
    \centering
    \includegraphics[width=0.92\textwidth]{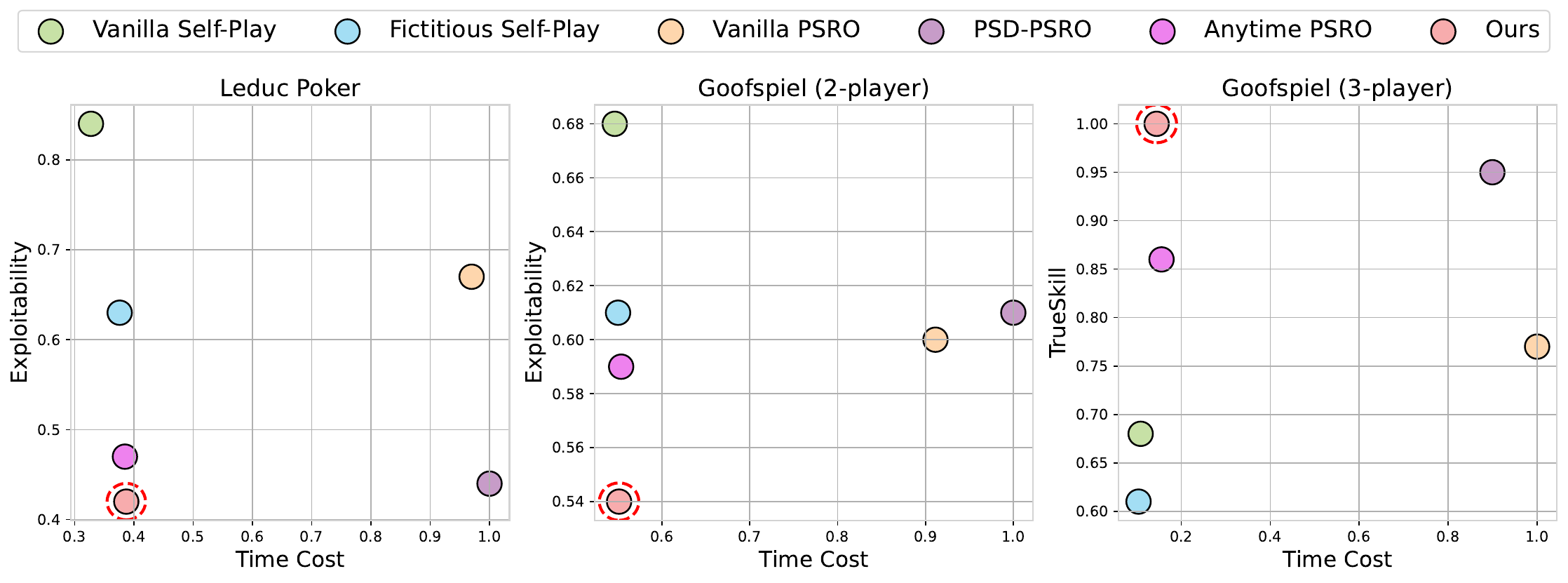}
    \caption{Exploitability v.s. time cost trade-offs across Leduc Poker and Goofspiel}
    \label{f5}
\end{figure*}

In this section, we primarily investigate problem (1). The exploitability \cite{timbers2020approximate} for each method in the Leduc Poker and Goofspiel are shown in Fig. \hyperref[f4]{4}. In Fig. \hyperref[f5]{5}, the horizontal and vertical axes represent the running time and performance of each method, respectively. To facilitate comparison, we normalize the running time on the horizontal axis.

Leduc Poker is a simplified variant of poker, featuring a deck with two suits and three cards per suit. Each player antes one chip and is dealt a single private card. In Leduc Poker, our method is implemented based on the Dynamic Window mechanism and combined with PSD, without incorporating MRCP\footnote{We find that incorporating MRCP into our method does not yield positive effects, possibly due to game-specific factors.}. To ensure a fair comparison, PSD was also included in the implementations of Fictitious Self-Play and Vanilla Self-Play methods. Goofspiel is a large-scale, multi-stage, simultaneous-move game, implemented in OpenSpiel. In Goofspiel, our method is implemented based on the Dynamic Window mechanism and combined with MRCP, without incorporating PSD\footnote{We find that incorporating PSD into our method does not yield positive effects, possibly due to game-specific factors.}. 

As shown in Fig. \hyperref[f4]{4}, in Leduc Poker, our method slightly outperforms PSD-PSRO and significantly outperforms the other methods. Similarly, in Goofspiel (2-player), our method outperforms the other methods. In Fig. \hyperref[f5]{5}, our method consistently lies on the Pareto frontier \cite{yang2022review} across all three games. It outperforms the best-performing baselines in each game—PSD-PSRO in Leduc Poker, Anytime PSRO in Goofspiel (2-player), and PSD-PSRO in Goofspiel (3-player)—in terms of both performance and running time. Although our method involves computing Nash clustering, which leads to slightly higher running time compared to Fictitious Self-Play and Vanilla Self-Play, its performance is substantially superior to both. This indicates that our method strikes a favorable trade-off between performance and efficiency across different games. Moreover, the results on Goofspiel (3-player) highlight that the efficiency gains are even more pronounced in multi-player scenarios, as the ratio $\frac{\text{Time Cost (PSRO)}}{\text{Time Cost (Ours)}}$ (6.23) is larger than that in both Goofspiel (2-player) (1.81) and Leduc Poker (2.58).

\subsection{Ablation Study}

\begin{figure*}[htbp]  
    \centering
    \includegraphics[width=0.75\textwidth]{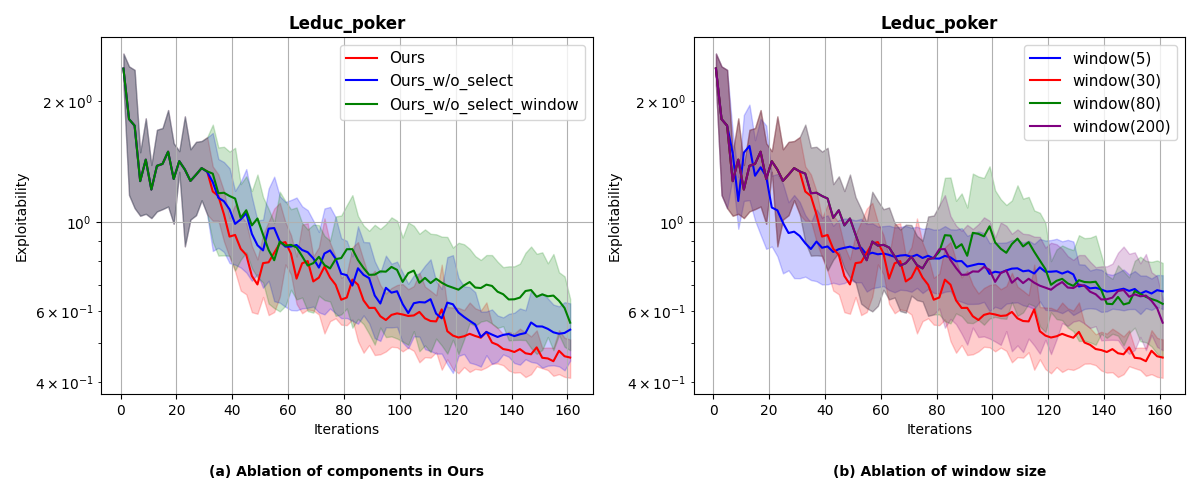}
    \caption{Ablations on key components and window size within Dynamic Window mechanism}
    \label{f6}
\end{figure*}

In this section, we investigate problem (2). In Leduc Poker, we compare the impact of each component in the Dynamic Mechanism on the final performance, as shown in Fig. \hyperref[f6]{6(a)} ("Ours\_w/o\_window"  is not included because without a window size restriction, there is no need for eliminated strategy selection). Additionally, the impact of the window size on final performance is compared in Fig. \hyperref[f6]{6(b)}.

As shown in Fig. \hyperref[f6]{6(a)}, omitting either the eliminated strategy selection component or the window size restriction component significantly degrades performance. For the former, randomly selecting eliminated strategies leads to the removal of some critical strategies, highlighting the importance of how eliminated strategies are selected. For the latter, not restricting the window size introduces the two challenges we discuss in Section 3.1, further validating that our method effectively addresses these challenges. As shown in Fig. \hyperref[f6]{6(b)}, the window size influences the final performance. If the window size is too small, some critical strategies may not be included in the window, while if it is too large, the aforementioned challenges arise again.

\subsection{Compatibility Analysis}

In this section, we focus on investigating problem (3). As summarized in Section 2, existing PSRO variants supporting simulation-free can be categorized from two perspectives: MSS and BRS. Among them, MRCP is the principle method for MSS in SF-PSRO, while BD is the principle method for BRS in SF-PSRO. In our experiments, we select Anytime PSRO as the representative method for MRCP, and adopt PSD—the diversity regularization term—as the representative for BD, applied to the BRS of Fictitious Self-Play.

Fig. \hyperref[f7]{7} reports the performance comparison of BD and MRCP before and after incorporating our Dynamic Window. As shown in Fig. \hyperref[f7]{7(a)} and Fig. \hyperref[f7]{7(b)}, after integration, both PSD and MRCP achieve consistent performance improvements. "MRCP+Ours" shows only marginal improvement over "MRCP" in Leduc Poker, possibly because MRCP already performs well in this environment, leaving limited room for additional gains. In contrast, "MRCP+Ours" significantly outperforms "MRCP" in Goofspiel, as detailed in Appendix C. These results demonstrate that the proposed Dynamic Window integrates well with existing PSRO variants.

\begin{figure}[htbp]  
    \centering
    \includegraphics[width=\columnwidth]{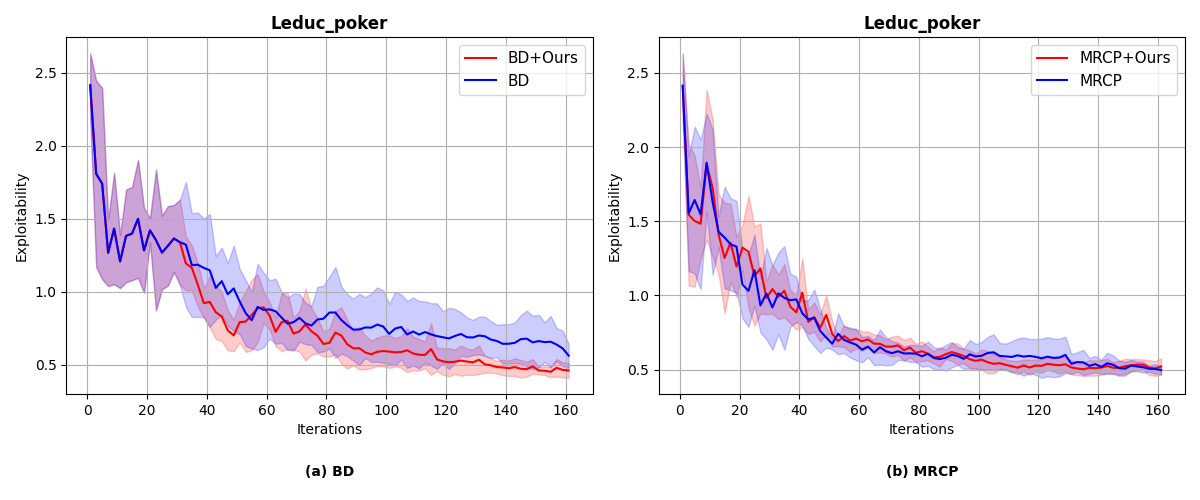}
    \caption{Compatibility of Dynamic Window mechanism with existing PSRO variants}
    \label{f7}
\end{figure}

\subsection{Additional Findings}

Since PSRO involves running a large number of iterations, and each iteration requires computing best-responses via RL, the reset of the optimizer for the policy/value networks in RL becomes an important consideration. We observe that most existing open-source implementations overlook this aspect, often reusing the optimizer across iterations. Fig. \hyperref[f8]{8} shows that the reset of the optimizer at each iteration generally leads to better performance for PSRO variants. This improvement may be attributed to the fact that carrying over historical information from previous iterations can negatively affect the training dynamics in subsequent iterations.

\begin{figure}[htbp]  
    \centering
    \includegraphics[width=\columnwidth]{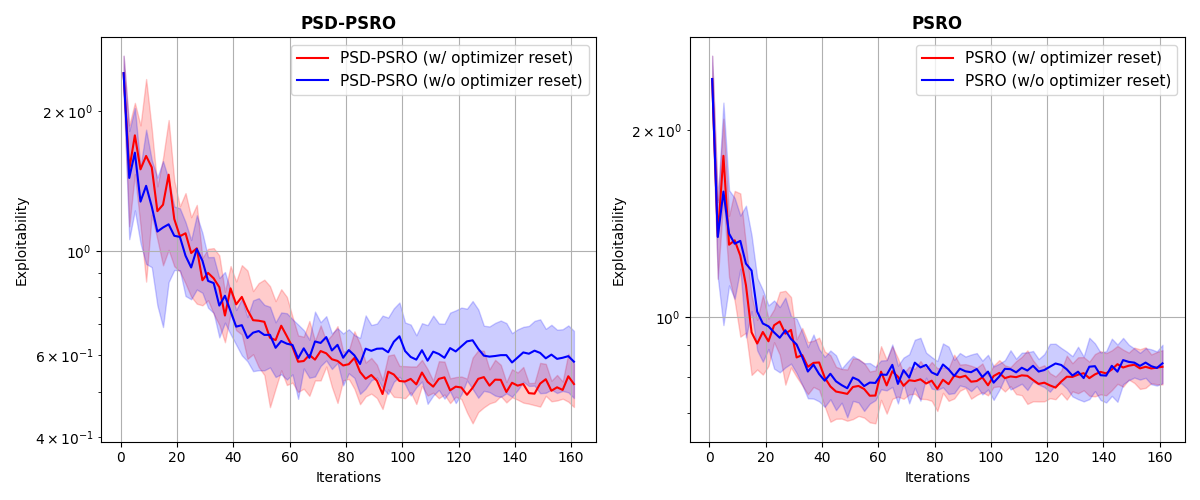}
    \caption{Impact of optimizer reset on best-response within PSRO variants}
    \label{f8}
\end{figure}


\section{Conclusions}

In this paper, we identify Game Simulation as the primary source of PSRO's high computational cost and introduce Simulation-Free PSRO (SF-PSRO) to eliminate this bottleneck. We review existing SF-PSRO variants and further propose a novel Dynamic Window-based SF-PSRO, which maintains a limited strategy window to simplify opponent selection and enhance best-response robustness. Experiments show that our method achieves competitive performance with substantially reduced time cost, validates the effectiveness of each component, and demonstrates strong plug-and-play compatibility with existing SF-PSRO variants.

Despite achieving competitive performance with significantly reduced computational time, SF-PSRO variants still suffer from several limitations.
First, in Dynamic Window-based SF-PSRO, a critical hyperparameter is the window size. Although our ablation studies have shown that this parameter should neither be too large nor too small—with a value of 30 proving empirically suitable for games like Leduc Poker and Goofspiel—identifying the optimal window size for more complex games may require additional, non-trivial tuning efforts.
Second, because SF-PSRO avoids explicit game simulations, it cannot construct a complete and accurate meta-payoff matrix. However, in standard PSRO, the final strategy is typically derived by computing a Nash equilibrium over this matrix during exploitability evaluation. Consequently, determining the final strategy in SF-PSRO becomes a non-trivial challenge. In Dynamic Window-based SF-PSRO, we mitigate this issue by maintaining a sketchy meta-payoff matrix. While less accurate than its fully simulated counterpart, this approximate matrix still provides a reasonable surrogate and can partially fulfill the role of the true meta-payoff matrix in guiding strategy selection.
We believe SF-PSRO represents a promising and valuable research direction. We look forward to future work that further refines and enhances SF-PSRO, making it more robust, adaptive, and broadly applicable across diverse game environments.




\begin{acks}
If you wish to include any acknowledgments in your paper (e.g., to 
people or funding agencies), please do so using the `\texttt{acks}' 
environment. Note that the text of your acknowledgments will be omitted
if you compile your document with the `\texttt{anonymous}' option.
\end{acks}



\bibliographystyle{ACM-Reference-Format} 
\bibliography{sample}



\appendix

\newpage

\section{Algorithms}

\label{a1}

Anytime PSRO \cite{mcaleer2022anytime}, Self-Play PSRO \cite{mcaleer2024toward}, and Efficient PSRO \cite{zhou2022efficient} correspond to Algorithm 1, Algorithm 2, and Algorithm 3, respectively.

\begin{algorithm}
	\caption{Anytime PSRO}
	\begin{algorithmic}[htbp]
		\State \textbf{Input:} initial strategy sets $X=(X_1, X_2)$
		\While{not terminated}
		\State Initialize the meta-strategy $\sigma_i$ to uniform over $X_i$ for $i \in \{1,2\}$
		\For{$i \in \{1, 2\}$}
		\For{$n$ iterations}
		\For{$m$ iterations}
		\State Update policy $\beta_{-i}$ against $\beta_{i} \sim \sigma_{i}$
		\EndFor
		\State Update $\sigma_i$ via regret minimization v.s. $\beta_{-i}$ (e.g., via Equation (1))
		\EndFor
		\State $X_i \gets X_i \cup \{\beta_i\}$ for $i \in \{1,2\}$
		\EndFor
		\EndWhile
		\State \textbf{Return:} $\sigma$
	\end{algorithmic}
\end{algorithm}

\begin{algorithm}
	\caption{Self-Play PSRO}
	\begin{algorithmic}[H]
		\State \textbf{Input:} initial strategy sets $X=(X_1, X_2)$
		\While{not terminated}
		\State Initialize new strategy $\nu_i$ arbitrarily
		\State Initialize the meta-strategy $\sigma_i$ to uniform over $X_i$ for $i \in \{1,2\}$
		\For{$i \in \{1, 2\}$}
		\For{$n$ iterations}
		\For{$m$ iterations}
		\State Update policy $\beta_{-i}$ against $\beta_{i} \sim \sigma_{i}$
		\State Update new strategy $\nu_i$ against $\beta_{-i}$
		\EndFor
		\State Update $\sigma_i$ via regret minimization v.s. $\beta_{-i}$ (e.g., via Equation (1))
		\EndFor
		\State $X_i \gets X_i \cup \{\beta_i,\nu_i\}$ for $i \in \{1,2\}$
		\EndFor
		\EndWhile
		\State \textbf{Return:} $\sigma$
	\end{algorithmic}
\end{algorithm}

\begin{algorithm}
	\caption{Efficient PSRO}
	\begin{algorithmic}[!h]
		\State \textbf{Input:} initial strategy sets $X=(X_1, X_2)$
		\While{not terminated}
		\For{$i \in \{1, 2\}$ in parallel}
		\For{loop all active best response $\beta_i^j \in X_i$}
		\For{all $X_{-i}^{<j}$ in parallel} 
		\State $\beta_i^j,\sigma_{-i}^{<j} = \mathrm{SOLVEURR}(\beta_i^j, X_{-i}^{<j})$ (Similar to line 5-8 in Anytime PSRO)
		\If{the lowest $\beta_i^j$ meets stops condition}
		\State set it to fixed and $X_i = X_i \cup \{\beta_i^j\}$
		\State Generate a new active strategy
		\EndIf
		\EndFor
		\EndFor
		\EndFor
		\EndWhile
		\State \textbf{Return:} $\sigma$
	\end{algorithmic}
\end{algorithm}

\section{Benchmark and Implementation Details}

\label{a2}

\paragraph{Leduc Poker}

In Leduc Poker("leduc\_poker(player=2)" in OpenSpiel), we use a two-player setup. We apply the PSRO framework with a Meta-Nash solver, employing DQN as the oracle agent. The specific hyper-parameters used for this setup are listed in Table 1.

\begin{table}[htbp]
	\centering
	\caption{Hyper-parameters for Leduc Poker}
	\begin{tabular}{l|l}
		\hline
		\textbf{Hyperparameters} & \textbf{Value} \\
		\hline
		\textbf{Oracle} & \\
		\hline
		Oracle agent & DQN \\
		Replay buffer size & $10^4$ \\
		Mini-batch size & 512 \\
		Optimizer & Adam \\
		Learning rate & $5 \times 10^{-3}$ \\
		Discount factor ($\gamma$) & 1 \\
		Epsilon-greedy Exploration ($\epsilon$) & 0.05 \\
		Target network update frequency & 5 \\
		Policy network & MLP (64-64-64) \\
		Activation function in MLP & ReLU \\
		\hline
		\textbf{Vanilla PSRO} & \\
		\hline
		Episodes for each BR training & $2 \times 10^4$ \\
		Learning steps for BR training & 100 \\
		Meta-policy solver & Nash \\
		\hline
		\textbf{PSD-PSRO} & \\
		\hline
		Episodes for each BR training & $2 \times 10^4$ \\
		Learning steps for BR training & 100 \\
		Meta-strategy solver & Nash \\
		Diversity weight ($\lambda$) & 1 \\
		\hline
		\textbf{Dynamic Window-based SF-PSRO} & \\
		\hline
		Window size & 30 \\
		\hline
	\end{tabular}
\end{table}

\paragraph{Goofspiel}

In Goofspiel("goofspiel (player=2, num\_cards=5,             points\_order=descending, return\_type=win\_loss)" in OpenSpiel), we use a two-player, 5-card setup. We adopt a descending order, meaning the cards are bid in the sequence 5, 4, 3, 2, 1. Regarding the return, only the win/loss outcome is considered, with 1 for a win and 0 for a loss. We apply the PSRO framework with a Meta-Nash solver, using DQN as the oracle agent. Hyper-parameters are shown in Table 2.

\begin{table}[htbp]
	\centering
	\caption{Hyper-parameters for Goofspiel}
	\begin{tabular}{l|l}
		\hline
		\textbf{Hyperparameters} & \textbf{Value} \\
		\hline
		\textbf{Oracle} & \\
		\hline
		Oracle agent & DQN \\
		Replay buffer size & $10^4$ \\
		Mini-batch size & 512 \\
		Optimizer & Adam \\
		Learning rate & $5 \times 10^{-3}$ \\
		Discount factor ($\gamma$) & 1 \\
		Epsilon-greedy Exploration ($\epsilon$) & 0.05 \\
		Target network update frequency & 5 \\
		Policy network & MLP (128-128-128) \\
		Activation function in MLP & ReLU \\
		\hline
		\textbf{Vanilla PSRO} & \\
		\hline
		Episodes for each BR training & $3 \times 10^4$ \\
		Learning steps for BR training & 100 \\
		Meta-policy solver & Nash \\
		\hline
		\textbf{PSD-PSRO} & \\
		\hline
		Episodes for each BR training & $3 \times 10^4$ \\
		Learning steps for BR training & 100 \\
		Meta-strategy solver & Nash \\
		Diversity weight ($\lambda$) & 1 \\
		\hline
		\textbf{Dynamic Window-based SF-PSRO} & \\
		\hline
		Window size & 30 \\
		\hline
	\end{tabular}
\end{table}

\paragraph{Experiments Compute Resources}

In this paper, all experiments are conducted on Intel(R) Core(R) CPU i9-10900k @ 3.70GHz processors.

\section{Additional Experiments}

Despite achieving competitive performance while consuming significantly less time, SF-PSRO has certain limitations. The most notable limitation is the lack of game simulation, which prevents the construction of a meta-payoff matrix. Consequently, we are unable to compute the Nash equilibrium strategy for the current strategy set. However, exploitability in PSRO is typically evaluated based on the Nash equilibrium strategy, since it often yields the lowest exploitability. Without access to the Nash equilibrium, SF-PSRO variants usually resort to alternative strategies as the final output, which may lead to higher exploitability (It is worth noting that in our experimental evaluation, all methods compute exploitability based on their Nash equilibrium strategies to fairly assess the true performance of the resulting strategy sets). As a SF-PSRO method, our proposed Dynamic window-based SF-PSRO also encounters this issue. Nevertheless, we maintain a sketchy meta-payoff matrix—although it is less accurate than a fully simulated one, it still serves as a reasonable approximation and can partially fulfill the role of the true meta-payoff matrix.

\begin{figure}[htbp]  
	\centering
	\includegraphics[width=0.45\textwidth]{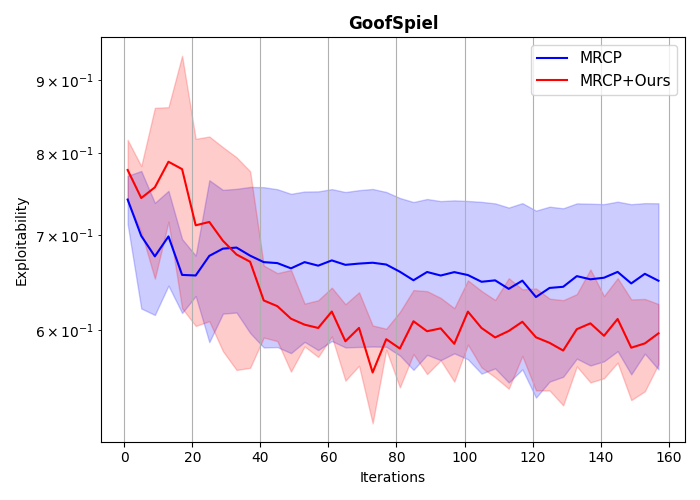}
	\caption{Compatibility of Dynamic Window mechanism with MRCP in GoofSpiel}
	\label{f3}
\end{figure}

\section{Convergence Proof}

\begin{theorem}[Convergence of PSRO under a Bounded Strategy Set]
	When the size of each player's strategy set is bounded by a constant $K$, the PSRO algorithm converges to a fixed point $(S_1^*, S_2^*, \dots, S_n^*)$ in a finite number of iterations.
\end{theorem}

\begin{proof}
	The proof proceeds as follows.
	
	\begin{enumerate}
		\item \textbf{Finiteness of Strategy Sets.} Since $|S_i^t| \leq K$ for all iterations $t$ and all players $i$, and each $S_i^t$ is a finite subset of player $i$'s full strategy space $\Sigma_i$, the total number of distinct joint strategy set configurations $(S_1, S_2, \dots, S_n)$ is finite.
		
		\item \textbf{Monotonicity via a Potential Function.} Define the potential function 
		\[
		\Phi(t) = \sum_{i=1}^n u_i(\mu_i^t, \mu_{-i}^t),
		\]
		where $\mu^t = (\mu_1^t, \dots, \mu_n^t)$ is the Nash equilibrium of the empirical game defined by $(S_1^t, \dots, S_n^t)$.
		
		\item \textbf{Effect of Adding Best Responses.} When player $i$ adds a best response $\mathrm{BR}_i(\mu_{-i}^t)$ to $S_i^t$, the resulting utility satisfies
		\[
		u_i\big(\mathrm{BR}_i(\mu_{-i}^t), \mu_{-i}^t\big) \geq u_i(\mu_i^t, \mu_{-i}^t).
		\]
		Consequently, the potential function $\Phi(t)$ is non-decreasing over iterations.
		
		\item \textbf{Strategy Pruning Rule.} We adopt the following pruning strategy: \emph{remove only those strategies that have zero probability in the current Nash equilibrium $\mu^t$}. This operation does not alter the equilibrium $\mu^t$ nor the associated payoffs, because strategies with zero probability do not affect expected utilities.
		
		\item \textbf{Convergence.} Combining the above:
		\begin{itemize}
			\item The number of possible joint strategy set combinations is finite.
			\item The potential function $\Phi(t)$ is non-decreasing and bounded above (since utilities are bounded in finite games).
			\item Each iteration either strictly increases $\Phi(t)$ or leaves it unchanged while possibly reducing the size of some $S_i^t$ via safe pruning.
		\end{itemize}
		Therefore, the algorithm must terminate after a finite number of steps at a fixed point $(S_1^*, \dots, S_n^*)$, where for every player $i$,
		\[
		\mathrm{BR}_i(\mu_{-i}^*) \in S_i^*,
		\]
		and no further strategies are added or removed. At this point, $\mu^*$ constitutes an exact Nash equilibrium of the restricted game, and the exploitability is zero within the current strategy subspace.
	\end{enumerate}
\end{proof}

Theorem~1 describes a relatively idealized scenario. In contrast, our Dynamic Window-based SF-PSRO faces two practical gaps when selecting strategies for removal:

(1) It may be impossible to identify strategies that have zero probability in the current Nash equilibrium—i.e., all strategies in the window are assigned strictly positive probability.

(2) Due to the absence of full game simulation, the sketchy meta-payoff matrix used for pruning introduces approximation error relative to the true meta-payoff.

Regarding (1), this issue becomes significantly less pronounced as the window size increases. Indeed, our experiments show that a window size of 30 consistently yields substantially better performance than a window size of 5. 

As for (2), since we only use the sketchy meta-payoff to identify a \emph{relatively weak} strategy (rather than an exact best response), the requirement on its accuracy is modest. Moreover, our experimental results validate that the sketchy meta-payoff is sufficiently reliable for effective strategy selection.

\end{document}